\documentclass[11pt]{article}

\usepackage{microtype}
\usepackage[numbers]{natbib}

\usepackage[utf8]{inputenc}
\usepackage[ruled,vlined]{algorithm2e}
\usepackage{natbib}
\usepackage{graphicx}
\usepackage{xspace}
\usepackage{amssymb}
\usepackage{amsmath}
\usepackage{amsthm}
\usepackage{dsfont}
\usepackage{float}
\usepackage{enumerate}
\usepackage{makecell}

\newtheorem{theorem}{Theorem}
\newtheorem{property}[theorem]{Property}
\newtheorem{requirement}[theorem]{Requirement}

\newtheorem{lemma}[theorem]{Lemma}
\newtheorem{definition}[theorem]{Definition}
\newtheorem{example}[theorem]{Example}
\newtheorem{corollary}[theorem]{Corollary}

\usepackage{url}

\usepackage{mathrsfs} 

\usepackage{natbib}
\usepackage{graphicx}
\usepackage{multirow}

\newcommand{\good}{\mathcal{F}}
\newcommand{\univ}{\kern0.75pt\mathcal{U}}
\newcommand{\oursys}{commutable\xspace}

\newcommand{\roots}{\textsc{roots}\xspace}

\newcommand{\goodtime}{\mathscr{G}_T}
\newcommand{\goodspace}{\mathscr{G}_S}
\newcommand{\restrtime}{\mathscr{R}_T}
\newcommand{\restrspace}{\mathscr{R}_S}
\newcommand{\restrbound}{\mathscr{R}_B}

\newcommand{\lev}{\textsc{lay}}
\newcommand{\level}{layer\xspace}
\newcommand{\levels}{layers\xspace}

\DeclareMathOperator*{\argmin}{argmin}

\newcommand{\source}{\textsc{source}\xspace}
\newcommand{\start}{\textsc{start}\xspace}
\newcommand{\parent}{\textsc{parent}\xspace}
\newcommand{\parind}{\textsc{pi}\xspace}
\newcommand{\core}{\textsc{core}\xspace}
\newcommand{\restr}{\textsc{restr}\xspace}

\newcommand{\children}{\textsc{children}\xspace}

\newcommand{\rec}{\textsc{spawn}\xspace}
\newcommand{\comp}{\textsc{complete}\xspace}

\newcommand{\choice}{\textsc{choose}\xspace}
\newcommand{\erre}{\textsc{r}\xspace}

\newcommand{\mcs}{\textsc{mcs}\xspace}
\newcommand{\mcss}{\textsc{mcs}s\xspace}
\newcommand{\mccis}{\textsc{mccis}\xspace}
\newcommand{\mcciss}{\textsc{mccis}s\xspace}
\newcommand{\white}{white\xspace}
\newcommand{\black}{black\xspace}
\newcommand{\bclique}{\textsc{bc}-clique\xspace}
\newcommand{\bcliques}{\textsc{bc}-cliques\xspace}

\newtheorem{fact}[theorem]{Fact}
\usepackage{soul,color}
\soulregister\cite7
\soulregister\ref7
\soulregister\pageref7
\usepackage{xcolor}

\let\oldnl\nl
\newcommand{\nonl}{\renewcommand{\nl}{\let\nl\oldnl}}

\newcommand{\qexd}{\hfill\textcolor{darkgray}{$\blacktriangleleft$}}

\makeatother

\begin{document}

\sloppy
\title{Listing Maximal Subgraphs in Strongly Accessible Set Systems\thanks{Present address: Alessio Conte, Roberto Grossi, Andrea Marino, and Luca Versari at the
Dipartimento di Informatica, Universit\`a di Pisa,
\texttt{\{conte,grossi,marino,luca.versari\}@di.unipi.it}.}}
\author{Alessio~Conte
\and 
Roberto~Grossi
\and
Andrea~Marino
\and
Luca~Versari
}
\date{March 2018}

\maketitle

\begin{abstract}
 Algorithms for listing the subgraphs satisfying a given property (e.g.,~being a clique, a cut, a cycle, etc.) fall within the general framework of set systems.
 A set system $(\univ, \good)$ uses a ground set $\univ$ (e.g.,~the network nodes) and an indicator $\good \subseteq 2^{\univ}$ of which subsets of $\univ$ have the required property. For the problem of listing all sets in $\good$ maximal under inclusion, the ambitious goal is to cover a large class of set systems, preserving at the same time the efficiency of the enumeration. Among the existing algorithms, the best-known ones list the maximal subsets in time proportional to their number but may require exponential space. In this paper we improve the state of the art in two directions by introducing an algorithmic framework that, under standard suitable conditions, simultaneously (i)~extends the class of problems that can be solved efficiently to \emph{strongly accessible} set systems, and (ii)~reduces the additional space usage from exponential in $|\univ|$ to \emph{stateless}, thus accounting for just $O(q)$ space, where $q \leq |\univ|$ is the largest size of a maximal set in $\good$.
\end{abstract}

\section{Introduction}
\label{sec:introduction}

Algorithms for graph listing have a long history and even if they were born in the 70s in the context of enumerative combinatorics and computational complexity~\cite{goldberg1992efficient,lawler1980generating,read1981survey,stanton2012constructive,valiant1979complexity}, the interest has quickly broadened to a variety of other communities in computer science and not, massively involving algorithm design techniques. 

In network analysis discovering special communities corresponds to finding all the subgraphs with a given property~\cite{ahmed2015efficient,du2007community,lee2010survey,modani2008large,mokken1979cliques,tsourakakis2013denser}. In bioinformatics, listing all the solutions is desirable, as single or few solutions may not be meaningful due to noise of the data, noise of the models, or unclear objectives~\cite{ciriello2008review,klein2012structural,lacroix2008introduction,milreu2014telling,tanay2002discovering}. In graph databases, graph structures are used for semantic queries, where nodes, edges, and properties are used to represent and store data; retrieving information corresponds to find all the suitable subgraphs in these databases~\cite{angles2008survey,cohen2005abstract,williams2007graph}. When dealing with incomplete information, it may be impossible to completely satisfy a query. Subgraph listing algorithms can find answers that \emph{maximally} satisfy a partial query; for instance, there is a one-to-one correspondence between the results of a join or full disjunction query and certain subgraphs of the \emph{assignment graph}, a special graph obtained by combining the relational database with the given query. 
Moreover, the kind of subgraphs to look for depend not only on the database, but also on the query~\cite{cohen2005abstract}.

In this scenario, graph enumeration has left the theoretical border~\cite{catalogue} to meet more stringent requirements: not only a given listing problem must fit a given class of complexity, but its algorithms must be efficient also in real-world applications. To promote efficiency, many works try to eliminate redundancy by, for example, listing only solutions which respect properties of \emph{closure}\cite{boley2010listing} or \emph{maximality}\cite{cohen2008generating}. On the other hand, algorithm design has made a big effort to generalize the graph properties to be enumerated and to unify the corresponding approaches~\cite{avis1996reverse,chiba1985arboricity,cohen2008generating,lawler1980generating,tsukiyama1977new}. These generalizations allow the same algorithm to solve many different problems.

The contribution of this paper fits into this line of research: on one side, we want to obtain efficient listing algorithms able to list maximal solutions in large networks; on the other side, we aim at designing an algorithmic framework which solves simultaneously many problems and leave the designer in charge of few core tasks depending on the specific application. In particular, we focus on efficient enumeration algorithms for maximal subgraphs satisfying a given property (e.g.~being a clique, a cut, a cycle, a matching, etc.), as they fall within the general framework of set systems~\cite{cohen2008generating,lawler1980generating}.

\paragraph*{Set systems and maximality.} Consider for instance the problem of listing the maximal cliques in a graph $G = (V,E)$. We can see this problem as defined on a ground set $\univ = V$, where each clique is represented by its vertex set $X \subseteq \univ$ that satisfies the condition $X \in \good$, where membership in $\good$ is equivalent to say that for each pair $u,v \in X$ such that $u \neq v$, we have $\{u,v\} \in E$. It is also always assumed that $\emptyset \in \good$. 
In general,
a set system $(\univ, \good)$ is a unifying model to formulate listing problems on graphs by specifying a ground set $\univ$ of objects (e.g., the vertices of the graph) and indicating which combination of these objects meets the wanted property, expressed as the collection $\good\subseteq 2^{\univ}$ of ``good'' subsets that fulfill the property~\cite{lawler1980generating}. For this reason, in the remainder of the paper we will use the term ``property'' as a synonym of set system, meaning the set system induced by the given property (e.g., being a clique) in a graph. Moreover, we assume in the rest of the paper that $\good$ is given algorithmically, i.e. we can check whether $X\in\good$ in time $\goodtime = poly(|\univ|)$ and space $\goodspace = \Omega(q)$, where $q$ is the maximum size among the sets $X\in\good$. 

We are interested in listing just the \emph{maximal solutions} of $\good$, where a solution $X \in \good$ is maximal if there exists no $Y \in \good$ such that $X \subset Y$. This is of interest for many applications in literature (e.g.~\cite{chiba1985arboricity,tsukiyama1977new,JOHNSON1988119,koch1996algorithm}) where instances can have a number of maximal solutions that is much smaller than $|\good| \leq 2^{|\univ|}$. However, generating only maximal solutions usually makes the listing problem harder, as it generally requires more sophisticated techniques: for example, we are not aware of any \emph{Constant Amortized Time} algorithm for listing maximal patterns in graph, while many have been designed for non-maximal listing problems\cite{ruskey2003combinatorial,uno2015constant}. In some cases, it may even be difficult to decide whether a solution $X\in \good$ is maximal or not. In this paper we will focus on set systems which respect the following property: $X,Y\in\good$ and $X\subset Y$ implies that there exists $z\in Y\setminus X$ such that $X \cup\{z\}\in\good$~\cite{arimura2009polynomial,boley2010listing}. This class of set systems is called \emph{strongly accessible} and it is the most general one for which the maximality of a solution can be verified in polynomial time.

We remark that there is a trade-off between generality and efficiency, as a truly general output sensitive\footnote{We recall that output-sensitive means there is a bound for the overall time-cost of the algorithm which is linear in the number of solutions and polynomial in the size of the instance. This corresponds also to the notion of P-enumerability introduced by Valiant~\cite{valiant1979complexity}.} framework for listing maximal solutions seems hard to achieve for all set systems: it is known from the early 80s that this would imply \textsc{p}=\textsc{np}~\cite{lawler1980generating}. 
For this reason, Lawler et al.~\cite{lawler1980generating} and later Cohen et al.~\cite{cohen2008generating} have shown that the hardness of a listing problem can be linked to that of solving an easier core task, called in~\cite{cohen2008generating} \emph{restricted problem}.
Keeping this in mind, we will proceed along two orthogonal directions simultaneously: (i)~our framework can deal with a wide class of set systems; (ii)~we heavily improve the memory usage over the currently known exponential space (which means just a poly-logarithm of the output size~\cite{fukuda96}). We actually go further when output has not to be stored, and show that our framework is \emph{stateless}, meaning that it uses just $O(q)$ additional space, whenever the restricted problem allows for it.\footnote{We argue that this improvement affects the running time too, as this extremely small space is more likely to fit private caches in multiple core executions, where the massive graph is in the shared memory (e.g.~\cite{noi2016icalp}). 
However, a discussion about this topic is outside the scope of this paper.}

\begin{figure}[t]
    \centering
    \includegraphics[width=.6\textwidth]{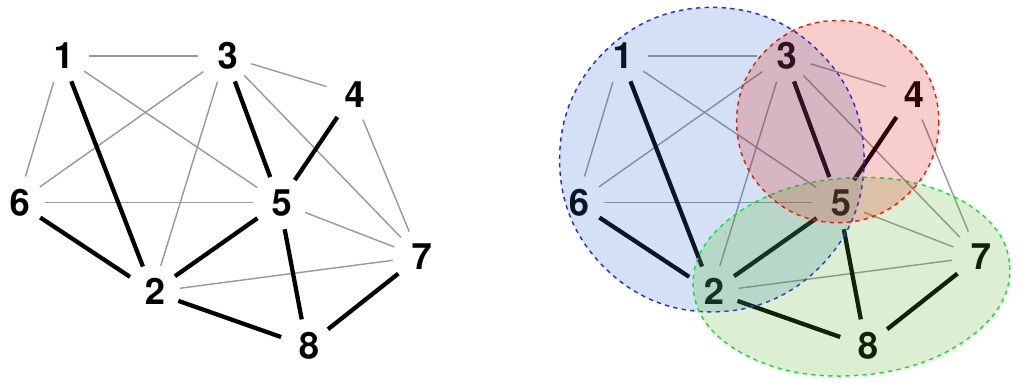}
    \caption{A graph with black and white edges, and its \bcliques}
    \label{fig:bccliques}
\end{figure}

\paragraph*{Case studies and open problems.} Before claiming our results, we give some necessary background for our work. Many works in the literature focus their attention on listing maximal solutions of specific set systems, solving them with ad hoc techniques. For example, they focus on cliques~\cite{bronkerbosch,EppsteinLS10,TomitaK09,makino2004new,Akkoyunlu73}, independent sets~\cite{JOHNSON1988119,tsukiyama1977new}, acyclic subgraphs~\cite{yau1967generation,schwikowski2002enumerating}, matchings~\cite{Gely20091447,Uno01}, $k$-plexes~\cite{berlowitz2015efficient,wu2007parallel} which all correspond to \emph{hereditary properties} (or more generally \emph{independence set systems}), i.e. $X\subseteq Y\subseteq \univ$ and $Y\in\good$ implies $X\in\good$~\cite{lawler1980generating}. These kinds of problems can be often efficiently solved in output sensitive time and polynomial space by using general frameworks~\cite{chiba1985arboricity,cohen2008generating,lawler1980generating}. 

However, some tricky properties do not fit in the above scenario. The \emph{maximal common connected induced subgraphs} (\mcciss) between two graphs $A$ and $B$ is one notable example. This problem is at least as hard as the graph isomorphism problem~\cite{babai2016graph}.
For this reason a weaker form is considered, where the individual isomorphisms are listed (noting that many isomorphisms can correspond to the same \mccis). Surprisingly, no output-sensitive listing algorithm with polynomial space is known for even this version, while several papers~\cite{Cao01082008,koch2001enumerating,koch1996algorithm} adapt existing techniques for maximal clique enumeration without any guarantee. We will use this problem as a running example through the paper, showing that our framework yields an output-sensitive and polynomial space algorithm for the problem. What is known is the transformation by Koch~\cite{koch1996algorithm}, so that it reduces to the enumeration of special cliques in a new graph, called \emph{product graph} $G$, whose edges are either black or white (see Figure~\ref{fig:bccliques}): 
let $G_B$ be the edge subgraph of $G$ containing only the black edges and define a \bclique as a connected subgraph of $G_B$ whose nodes form a clique in $G$; then, for each (maximal) isomorphism of $A$ and $B$, there is a (maximal) \bclique in $G$, and vice versa.

Listing \bcliques share with other listing problems the so called \emph{connected hereditary property}: all the solutions in $\good$ are connected in some given graph ($G_B$ in the case of \bcliques); moreover, given $X\subseteq Y\subseteq \univ$, if $X$ is connected and $Y\in\good$ then $X\in\good$, i.e., $\good$ is closed wrt connected induced subgraphs~\cite{cohen2008generating}. 
Finding a memory-efficient framework for connected hereditary graph properties is still open, as the only option is using the framework in~\cite{cohen2008generating,cohen2005abstract}, which requires keeping all the maximal solutions in memory. This open problem is interesting as its solution would allow memory-efficient enumeration also for other well-known problems such as connected $k$-plexes~\cite{berlowitz2015efficient} and several types of database queries~\cite{cohen2005abstract}.

We remark that hereditary and connected hereditary set systems are both also strongly accessible set systems, as shown in Figure~\ref{fig:hierarchy}. However, it is easy to show that the strongly accessible class contains more than just hereditary and connected hereditary properties: for example, if $R$ is a set of vertices, the property of ``being a \bclique containing a vertex of $R$'' is not connected hereditary (as a connected subset will still be a \bclique but may not contain any vertex from $R$), but it is straightforward to see that it is still strongly accessible. In general, we will refer to hereditary and connected hereditary properties with the additional constraint of containing a vertex from a required set as \textit{required hereditary} and \textit{required connected hereditary} respectively. Such properties are strongly accessible, but they are no more hereditary or connected hereditary.
For the sake of clarity, Table~\ref{tbl:classes} summarizes our definitions and gives some examples.

\begin{table}[t]
\renewcommand{\arraystretch}{2}
\centering
\begin{tabular}{|p{1in}|p{2in}|p{6cm}|}
    \hline
    \multicolumn{1}{|c|}{\textsc{class}} & \multicolumn{1}{c|}{\textsc{examples}} & \multicolumn{1}{c|}{\textsc{property}}\\
    \hline
    \multirow{1}[8]{1in}{Hereditary (ISS)}          & Cliques, Independent Sets, Induced Forests  & $X\subseteq Y\subseteq \univ$ and $Y\in\good$ implies $X\in\good$\\
    \hline
    \multirow{1}[10]{1in}{Connected Hereditary} & Connected \textit{k}-plexes, Common Connected Induced Subgraphs (\bcliques), Induced Trees & Let $X\subseteq Y\subseteq \univ$. If $X$ is connected and $Y\in\good$ then $X\in\good$\\
    \hline
    \multirow{1}[10]{1in}{Strongly Accessible}  & Hereditary, Connected Hereditary, Required Hereditary, Required Connected Hereditary & $X,Y\in\good$ and $X\subset Y$ implies that there exists $z\in Y\setminus X$ such that $X \cup\{z\}\in\good$\\
    \hline
\end{tabular}
\caption{Classes of Set Systems.}
\label{tbl:classes}
\renewcommand{\arraystretch}{1}
\end{table}

\paragraph*{Results.} In this paper, we provide a stateless output-sensitive framework for connected hereditary properties. This is the first polynomial space output-sensitive framework for this kind of systems, solving the open problem by~\cite{cohen2008generating}, as stateless implies polynomial space (actually it can be much less). As an application, our result implies the first output sensitive polynomial memory listing algorithm for maximal isomorphisms (\bcliques) corresponding to maximal common connected induced subgraphs, using just $O(q)$ memory. 

Our approach extends to the strongly accessible set systems independently from the core task, with a cost which is not far from the optimum, which we prove to be in the worst case $\Omega(\alpha 2^{q/2})$, where $\alpha$ is the number of maximal solutions, unless the Strong Exponential Time Hypothesis~\cite{impagliazzo1999complexity} is false. 
\begin{figure}
    \centering
    \includegraphics[width=.4\textwidth]{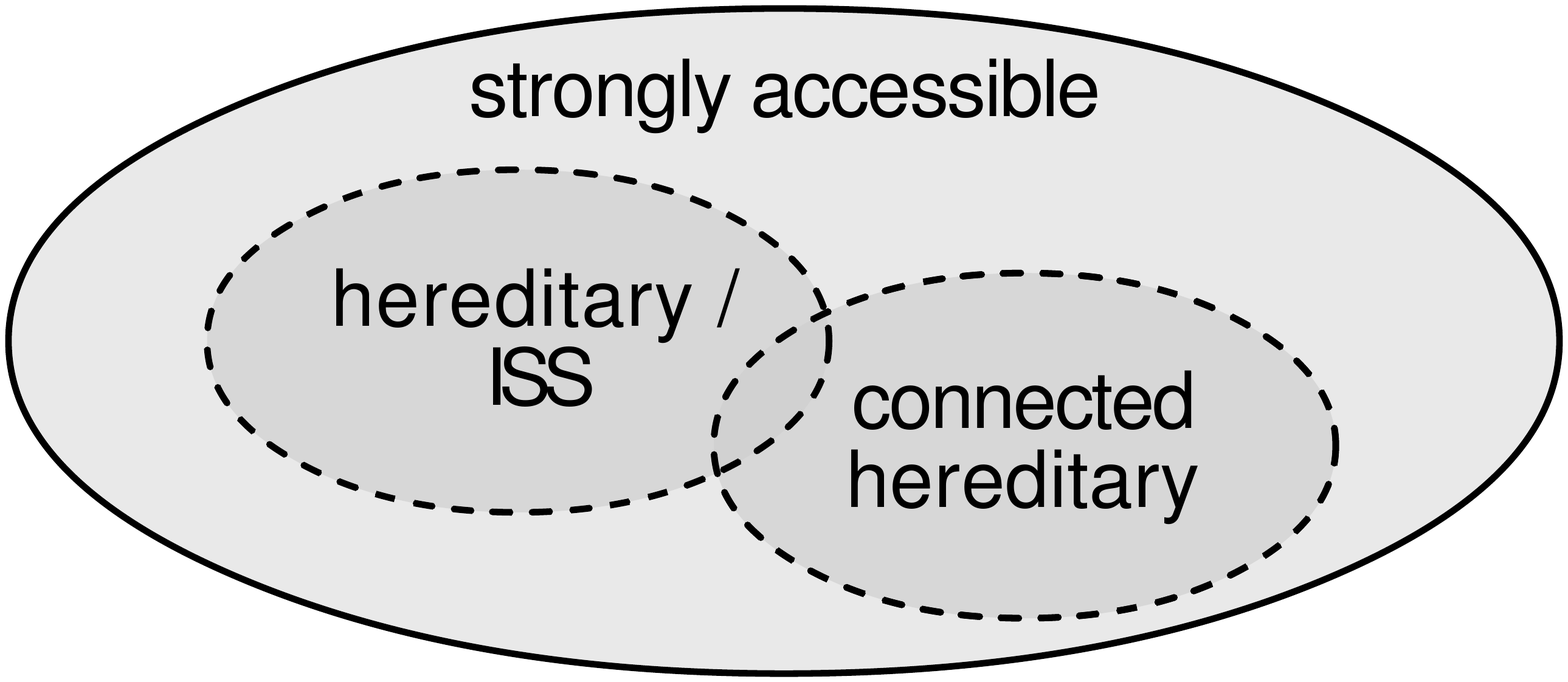}
    \caption{Hierarchy of Set Systems}
    \label{fig:hierarchy}
\end{figure}

\paragraph*{Techniques.}
When designing listing algorithms, \emph{binary partition} is a simple yet powerful technique which consists in recursively dividing the set of solutions into two or more non-overlapping partitions. In order to get output-sensitive algorithms, it is necessary to decide whether each of these sets is empty in polynomial time. However, when dealing with maximal solutions, a result by Lawler et al.~\cite{lawler1980generating} which we will discuss in Section~\ref{sec:extended-work} implies that this cannot be done even for restrictive classes of problems unless \textsc{p=np}.

In this scenario, reverse search by Avis and Fukuda~\cite{avis1996reverse} is a more powerful technique avoiding this kind of problems. Its general schema is based on operations which allow jumping from a solution to another. These operations define a graph-like structure as that represented in Figure~\ref{fig:solspace}, that can be traversed in order to find all of them. The key feature of reverse-search is then pruning this structure to obtain a tree-like (or forest-like) one: a tree-like structure is easy to traverse without incurring multiple times in the same nodes, whereas a more complex graph-like structure requires keeping in memory all the visited nodes. For the reader who is not familiar with reverse search, we refer to the original paper~\cite{avis1996reverse} or to our overview in Section~\ref{sec:classical-reverse-search}. 

Whether reverse search could be efficiently adapted to be output sensitive and memory efficient on non-hereditary set systems was an open problem~\cite{cohen2008generating}. We show that state of the art strategies for many specific set systems cannot be applied, as the routines to transform a maximal solution into another, when using the lexicographical order of the elements, involve NP-hard problems. We solve the problem by representing a solution not as \emph{a bag of elements}, but with a \emph{canonical} representation which takes into account the topology of the solution. This new layer of complexity needs non trivial ideas and technicalities to avoid generating duplicate solutions, in order to ensure correctness without storing the output in memory.

\begin{figure}
    \centering
    \includegraphics[width=0.5\textwidth]{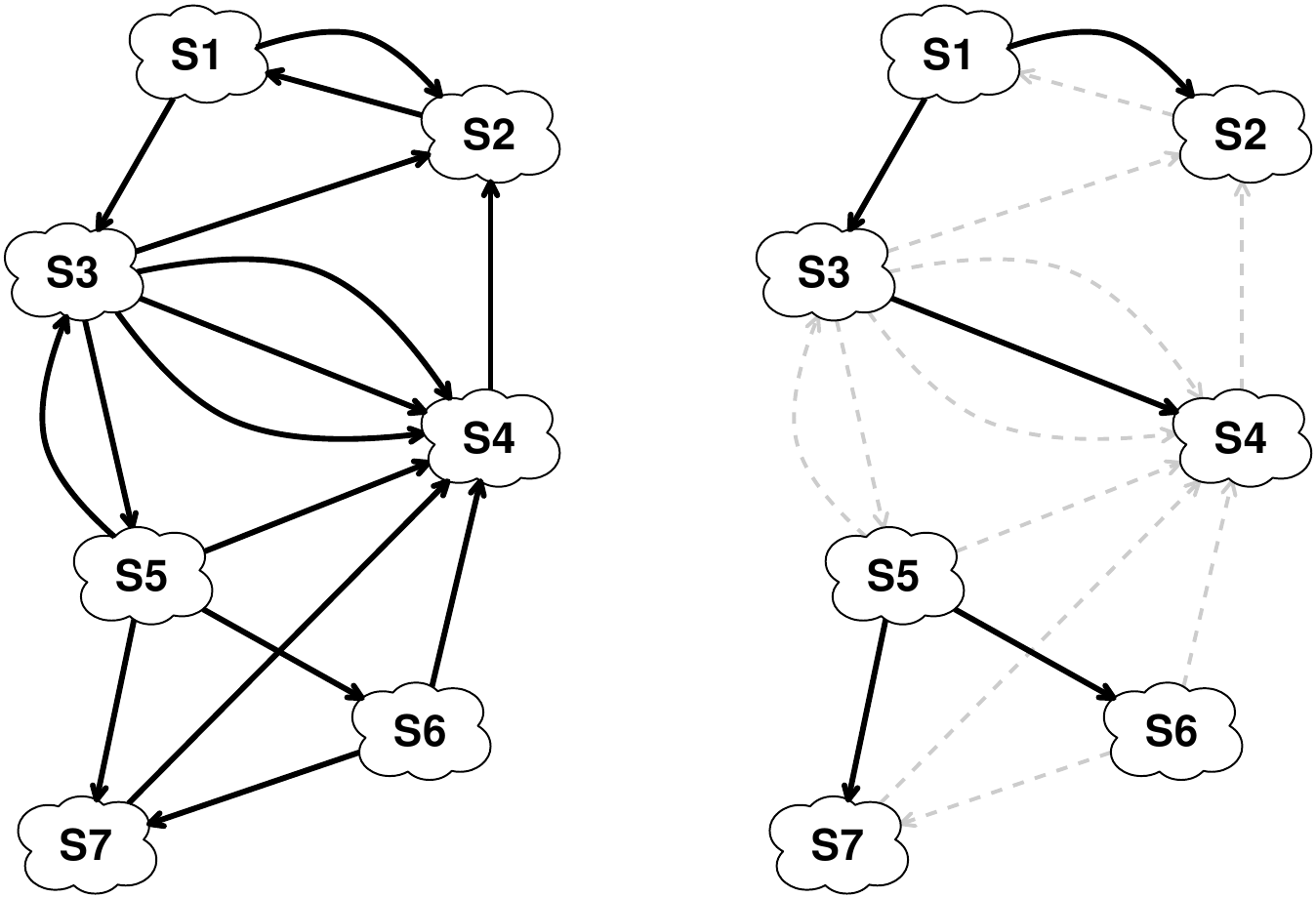}
    \caption{On the left, a graph-like structure whose nodes represent the maximal solutions of a set system and edges represent all the possible ways to ``jump'' from a solution to another. On the right, a possible forest-like structure induced by a reverse-seach algorithm.}
    \label{fig:solspace}
\end{figure}


\section{Related work}
\label{sec:extended-work}

\paragraph*{Hardness and restricted problem for hereditary properties.} Hereditary properties, also called Independence set systems, have been well-known and studied since the 80's. One of the most striking result is the one by Lawler et al.~\cite{lawler1980generating}, which shows that generating all the maximal solutions of an independent set system in an output-sensitive fashion implies \textsc{p}=\textsc{np}. This is proved by reducing SAT formulas with $N$ variables to appropriate set systems which have $N$ dummy maximal solutions, plus one maximal solution for each satisfying assignment of the corresponding formula. Moreover, in the same work the notion of \emph{restricted problem} has been introduced. As the restricted problem is a sub-problem of the general one, we have that if the general problem can be solved with some guarantee, the restricted one can be solved with the same guarantee. 
On the other hand, both~\cite{lawler1980generating} and~\cite{cohen2005abstract} show that the running time and space requirement of the general problem also depends on those of the restricted one so that if the latter ones are easier to solve then the former can potentially achieve output sensitivity. For hereditary properties,~\cite{cohen2008generating} proves that the general problem can be solved in polynomial space and polynomial delay if the restricted problem can be solved in polynomial space and polynomial total time.

\paragraph*{Framework by Cohen et al. for connected hereditary properties.}
The work by~\cite{cohen2008generating} proves significant insights for connected hereditary set systems, by providing an algorithm for the general problem which uses the restricted problem as a subroutine. They show that if the restricted can be solved in polynomial time (resp. incremental polynomial time, polynomial total time) then the general problem can be solved in polynomial delay (resp. incremental polynomial time, polynomial total time). This relationship extends that among hereditary properties and their restricted problems, at the cost of using exponential space.
Indeed, for connected hereditary properties, their framework essentially corresponds to a visit on a graph-like structure like that in Figure~\ref{fig:solspace}, storing in memory \emph{all} maximal solutions found so far. This clearly requires exponential space, as the number of maximal solutions can be exponential. Reducing the space usage to polynomial, while keeping the delay bounded, is left as an open problem in~\cite{cohen2008generating}. Our paper will also fill this gap.

Calling $\restrtime$ and $\restrspace$ respectively the time and space required to solve the restricted problem for the problem at hand, we show in Table~\ref{tbl:comparison-v3} a summary of the results achieved by the techniques in~\cite{lawler1980generating} and~\cite{cohen2005abstract}, as well as those of this work.

\paragraph*{Frameworks for closed objects by Boley et al. and Uno et al.} 
Other works have focused on set systems more general than the ones studied by~\cite{cohen2008generating}, like strongly accessible set systems, but they used different requirements for the objects to look for. This is the case of the framework by~\cite{boley2010listing} and~\cite{arimura2009polynomial}. In these works, a set system $(\univ,\good)$ is given together with a \emph{closure} operator $\sigma:\good\to\good$, which satisfies  extensivity, monotonicity, and idempotence properties.
In such a context, a set $F\in \good$ is called closed if $\sigma(F)=F$ and these works aim at listing all the $F\in \good$ which are closed. In principle, the concept of closure is similar to that of maximality, in that both allow the algorithm to ignore many solutions that are not significant in that their information is included in some other more important one, thus improving the significance of the output and the total running time. We remark, however, that this is a different kind of problem as we require maximal solutions instead, for which monotonicity does not apply.

\begin{table}[t]
\centering
\begin{tabular}{|c|c|c|c|c|}
    \hline
    \textsc{class} & &
    \multicolumn{1}{c|}{\textsc{Lawler et al}~\cite{lawler1980generating}} &
    \multicolumn{1}{c|}{\textsc{Cohen et al}~\cite{cohen2008generating}} & \multicolumn{1}{c|}{\textsc{this work}}\\ 
    \hline
    \hline
    \multirow{2}{*}{Hereditary/ISS}  & \textsc{space} & \multicolumn{2}{c|}{$poly(|\univ|+\restrspace)$} & $O(\goodspace+\restrspace)$ \\
    \cline{2-5}
    & \textsc{delay} & \multicolumn{3}{c|}{$poly(|\univ|+\restrtime)$} \\
    \hline
    \hline
    \multirow{2}{*}{Connected Hereditary}  & \textsc{space} &  & $exp(|\univ|)$ & $O(\goodspace+\restrspace)$ \\
    \cline{2-2}\cline{4-5}
    & \textsc{delay} &  & \multicolumn{2}{c|}{ $poly(|\univ|+\restrtime)$ } \\
    \hline
    \hline
    \multirow{2}{*}{Strongly Accessible}  & \textsc{space}      & \multicolumn{2}{c|}{ } & $O(\goodspace)$    \\
    \cline{2-2}\cline{5-5}
    & \textsc{time}     & \multicolumn{2}{c|}{ } & $O(2^q)$    \\
    \hline
\end{tabular}
\caption{State of the art comparison, where $\goodtime$ and $\goodspace$ are respectively the time and space to recognize whether $X\in\good$, while $\restrtime$ and $\restrspace$ are the time and space required to get the next solution of the restricted problem, $q$ is the maximum size of the solutions in $\good$.}
\label{tbl:comparison-v3}
\end{table}

\paragraph*{Low memory enumeration.}
Even when a listing algorithm is time-efficient, i.e. polynomial total time or polynomial delay, turning this algorithm in a polynomial space one can be challenging. 
Indeed, as the size of the output can be exponential, polynomial space means that the algorithm may only store up to a poly-logarithm of the size of the output. 
As observed by Fukuda, introducing the notion of ``compactness'', ``an enumeration algorithm may not have to store all output in the active memory'': indeed, ``some algorithms can simply dump each output object without storing it''~\cite{fukuda96}. 
On the other hand, it is sometimes hard to achieve both output sensitive time and polynomial space, for instance because all the previous solutions need to be stored to check whether a new solution has been already generated or not, like in the approach by Cohen \emph{et al}.~\cite{cohen2008generating}.
When the absence of duplication can be ensured by the algorithmic structure, however, each solution can be output and thrown away, so in these cases polynomial space is often sufficient. In some special cases, it is even possible to design output sensitive algorithms whose additional space is just the maximum size of the objects to be listed~\cite{noi2016icalp,Conte2017mis,avis1996reverse}. One of the main ingredients of these kind of results is a \emph{stateless} structure, which is the capability of a recursive algorithm to rebuild the recursion stack when returning from a nested call.
As most current output-sensitive approaches have a recursive nature, they can reach $\Theta(n)$ nesting levels even for sparse graphs. In these cases, using the stack should be avoided: indeed, just storing one memory word per recursion level takes the memory usage to $\Omega(n)$. We will show that our framework can be implemented in a stateless fashion, using as little as $O(q)$ additional space.

\paragraph*{Case study: maximal common connected induced subgraphs.} For any two given input graphs $A$ and $B$, a
subgraph $S$ of $A$ is in common with $B$ if $S$ is isomorphic to a subgraph of $B$: it is maximal if there is no other common subgraph that strictly contains it. The \emph{maximal} common subgraph
(\mcs) problem requires discovering all the \mcs's of $G$ and $H$. The \mcs problem can be constrained to \emph{connected} and \emph{induced} subgraphs (\mccis)~\cite{Cao01082008,koch2001enumerating,koch1996algorithm}, where the latter means that all the edges
of $G$ between nodes in the \mcs are mapped to edges of $H$, and vice versa. The connectivity constraint is important to remove redundant solutions corresponding to compatible combinations of disjoint connected subgraphs, while the induced constraint is used to reduce the search space while still preserving the significance of the result~\cite{Cao01082008}.

Finding an output sensitive algorithm with polynomial space for the \mccis implies a polynomial algorithm for the well-known graph isomorphism problem, whose current best bound is quasi-polynomial~\cite{babai2016graph}. A weaker version corresponds to finding maximal isomorphisms corresponding to \mcciss. Finding an output sensitive algorithm with polynomial space for listing these isomorphisms is open, as several works~\cite{Cao01082008,koch2001enumerating,koch1996algorithm} adapt existing techniques for maximal clique enumeration without any guarantee. Indeed, by applying the transformation by Levi~\cite{levi73note}, the maximal isomorphisms corresponding to \mcss reduces to maximal clique enumeration in a new graph, called \emph{product graph}. A variation of this product graph by~\cite{koch1996algorithm} allows to reduce the maximal isomorphism corresponding to \mcciss to listing special cliques. In particular, given $A$ and $B$, the vertices of the product graph $G$ obtained by $A$ and $B$, are the pairs in $V(A)\times V(B)$ and there is a \black\ (resp. \white) edge between $(u,x)$ and $(v,y)$ whether $\{u,v\}\in E(A)$ and $\{x,y\}\in E(B)$ (resp. $\{u,v\}\not\in E(A)$ and $\{x,y\}\not\in E(B)$). Intuitively, an edge of $G$ is \black\ whether the assignments are compatible and the corresponding vertices are both connected in $A$ and $B$. Let $G_B$ be the edge subgraph of $G$ containing only black edges. For each isomorphism corresponding to \mcciss of $A$ and $B$, there is a maximal clique in $G$ connected by \black\ edges in $G_B$, and vice versa. We call these subgraphs of $G$ as \bclique. Hence, in general, given a graph $G$ whose edges are either black or white, we will show how to list all the maximal {\bclique}s in $G$ in an output sensitive fashion without storing all the solutions.

\section{Enumeration Framework for Strongly Accessible Set System}
\label{sec:our-enum-framework}

Reverse search is a powerful enumeration tool that can be made stateless, so a natural question is in which cases it can be efficiently applied to list the maximal solutions\footnote{Non maximal solutions can also be listed using reverse search but this is outside the scope of the paper.} of set systems. We provide answers to this question by introducing our framework for strongly accessible set systems
As this task is rather more complex than expected, we provide the following road map to help the reader. 

\subsection{Roadmap}
\label{sub:road-map}

We begin in Section~\ref{sec:classical-reverse-search} by highlighting some properties (in Definition~\ref{prop:rev-search}) that allow a set system to have a reverse search algorithm listing its maximal solutions. Looking at the literature, these properties are adopted by many enumeration algorithms, which explicitly or implicitly\footnote{In cases such as~\cite{avis1996reverse} this logic is encoded in the element selection order in the computational tree.} rely on a routine, which we denote \comp, that is applied to a suitable portion of the current maximal solution to generate another maximal solution (Property~\ref{prop:class-rev-search}). 

Routine \comp is at the heart of these enumeration algorithms because it gives a method to generate new maximal solutions from known ones. Its properties can be used to avoid generating duplicate solutions, thus giving a way to ensure correctness without storing the output in memory. Unfortunately, keeping the fundamental properties of \comp valid for (\oursys) strongly accessible set systems can be hard (Lemmas~\ref{lem:hardness} and~\ref{lemma:conditional_lb}).

In Section~\ref{sec:our-framework} we avoid this issue by giving a different definition of \comp, based on a black box \choice to select the next element to be added to the current partial solution $S'$. This induces a canonical ordering of the elements of a maximal solution $S$, which is not necessarily the sorted order of its elements, but it is rather dictated by the order in which \choice extracts the elements to be added. Differently from before, \comp takes polynomial time (iff \choice takes polynomial time) and the canonical ordering guarantees that the vertices in any prefix of $S$ induce a (partial) solution $S'\in \good$. We use this to define a total order among the maximal solutions, which can be listed by the reverse search (Theorem~\ref{the:main-strongly}) when a suitable condition (Requirement~\ref{prop:pref-min}) is satisfied.

In Section~\ref{sec:commutable} we refine our approach when the strongly accessible set system is commutable. We define \choice so that the elements of $S$ are partitioned into layers (Definition~\ref{def:layer}): intuitively, the elements in the same layer can be chosen in any given order; instead, those in the next layer should be taken only after all the elements in the current layer have been consumed by \choice.
Although this mechanism of layers is significantly more complex, it gives a more refined way to generate new maximal solutions from the current one (Lemma~\ref{lemma:key}). Its strength lays in generalizing the previous approaches~\cite{cohen2008generating,lawler1980generating} based on restricted problems to the wider class of \oursys set systems. We obtain improved bounds wherever the restricted problem can be solved efficiently (Theorems~\ref{the:commutable-enumeration}, \ref{the:commutable-running}, and~\ref{the:space-bounds}). 

Finally, in Section~\ref{sec:stateless} we observe that our algorithms can be made stateless, simulating the traversal of the directed forest implicitly induced by these algorithms. This approach uses just $O(q)$ more space than solving the restricted problem.

\subsection{Using reverse search: initial obstacles}
\label{sec:classical-reverse-search}

From now on, we will assume that $\good$ is nonempty, as otherwise the listing problem is trivial. For a set system $(\univ, \good)$, we define \parent  as a partial function whose domain is the set of maximal solutions for $(\univ, \good)$: given a maximal solution $S$, either \parent returns another maximal solution $S' \neq S$ or is \emph{undefined}. In the latter case, we say that $S$ is a \emph{root}.

\begin{definition}
\label{prop:rev-search}
A set system $(\univ, \good)$ is called \emph{reverse searchable} if it admits a \parent function that satisfies the following sufficient conditions for any maximal solution $S$.
\begin{enumerate}[\tt (A)]
\item\label{rev-b} $\parent(S)$ can be computed from $S$ given the knowledge of $(\univ, \good)$; also, a compatible \emph{order} $\prec$ (that is, $\parent(S)\prec~S$) on the maximal solutions of $(\univ, \good)$ must exist.
\item\label{rev-c} $\roots =\{ X : X \textrm{\ is\ a\ root} \}$ can be computed given the knowledge of $(\univ, \good)$. 
\item\label{rev-d} $\children(S) = \{X :\parent(X)=S\}$ can be computed from $S$ given $(\univ, \good)$.
\end{enumerate}
\end{definition}

Given a set system $(\univ, \good)$ that is \emph{reverse searchable} for a given \parent function, all the maximal solutions for $(\univ, \good)$ can be enumerated with no duplicates. Indeed, conditions~\ref{rev-b}--\ref{rev-d} in Definition~\ref{prop:rev-search} define a directed forest whose nodes, i.e., the maximal solutions of $(\univ, \good)$, can be listed by simulating a traversal, as proved next. A visual example of this structure is shown in Figure~\ref{fig:solspace} (right), where $S1$ and $S5$ are the roots.

\begin{lemma}
\label{lem:basic}
Given a set system $(\univ, \good)$ that is \emph{reverse searchable} for a given \parent function, all the maximal solutions for $(\univ, \good)$ can be enumerated with no duplicates.
\end{lemma}
\begin{proof}
The algorithm to find all the maximal solutions works as follows. For each $S\in\roots$, invoke the following recursive procedure $\textsc{rec}(S)$. Given a solution $S$, $\textsc{rec}(S)$ simply calls $\textsc{rec}(S')$ for each $S'\in\children(S)$. This suffices to find all the solutions. Indeed, by contradiction, let $S$ be the minimum solution under $\prec$ that is not found. If $S$ is a root, it is found by Definition \ref{prop:rev-search}.\ref{rev-c}. Otherwise, there exists $P=\parent(S)$, which is found since $P \prec S$ by Definition \ref{prop:rev-search}.\ref{rev-b}. As $S\in \children(P)$, $S$ is then found by Definition \ref{prop:rev-search}.\ref{rev-d}. Moreover, each solution is found exactly once. As each solution $S$ has only one parent and there are no cycles by Definition \ref{prop:rev-search}.\ref{rev-b}, we have that there is only one $P$ such that $S\in\children(P)$. 
\end{proof}

To actually implement Definition~\ref{prop:rev-search} many algorithms rely (directly or indirectly) on the following $\comp$ function, whose domain is the whole $\good$, not only the maximal sets in $\good$ (Tsukiyama et al.~\cite{tsukiyama1977new} inspired many of these algorithms).

\begin{property}
\label{prop:class-rev-search}
A set system $(\univ, \good)$ is \emph{reverse searchable} if the following conditions are satisfied.
\begin{itemize}
\item The compatible order between maximal solutions, i.e. $\prec$, is the lexicographical order.
\item Given a solution $X\in \good$, not necessarily maximal, $\comp(X)$ returns the lexicographic minimum among all the maximal solutions containing $X$.
\item $\parent(S)$, if it exists, is defined as $\comp(X)$ for some $X \subset S$, with $\comp(X) \neq S$. 
\item A child $S' \in \children(S)$, if it exists, is defined as $S'=\comp(X \cup \{w\})$ for some $X \subset S$ and $w \in \univ$. 
\end{itemize}
\end{property}
\begin{proof}
 If we fail to compute $\parent(S)$, then $S$ is a root. Otherwise, $\parent(S)=\comp(X)\neq S$ for some $X \subset S$, and the above conditions imply that $\comp(X) \prec S$ as $S$ is one of the maximal solutions containing $X$ but it cannot be the lexicographic minimum. Finally, $ \children(S)$ can be computed using \comp as stated in the last condition.
\end{proof}

The conditions in Property~\ref{prop:class-rev-search} imply those in Definition~\ref{prop:rev-search} (the converse is not necessarily true). Unfortunately this approach cannot be applied to strongly accessible set systems, as the $\comp$ function in Property~\ref{prop:class-rev-search} can be NP-hard to compute, e.g. in the case of \bcliques.

\begin{figure}[t]
    \centering
    \includegraphics[scale=0.35]{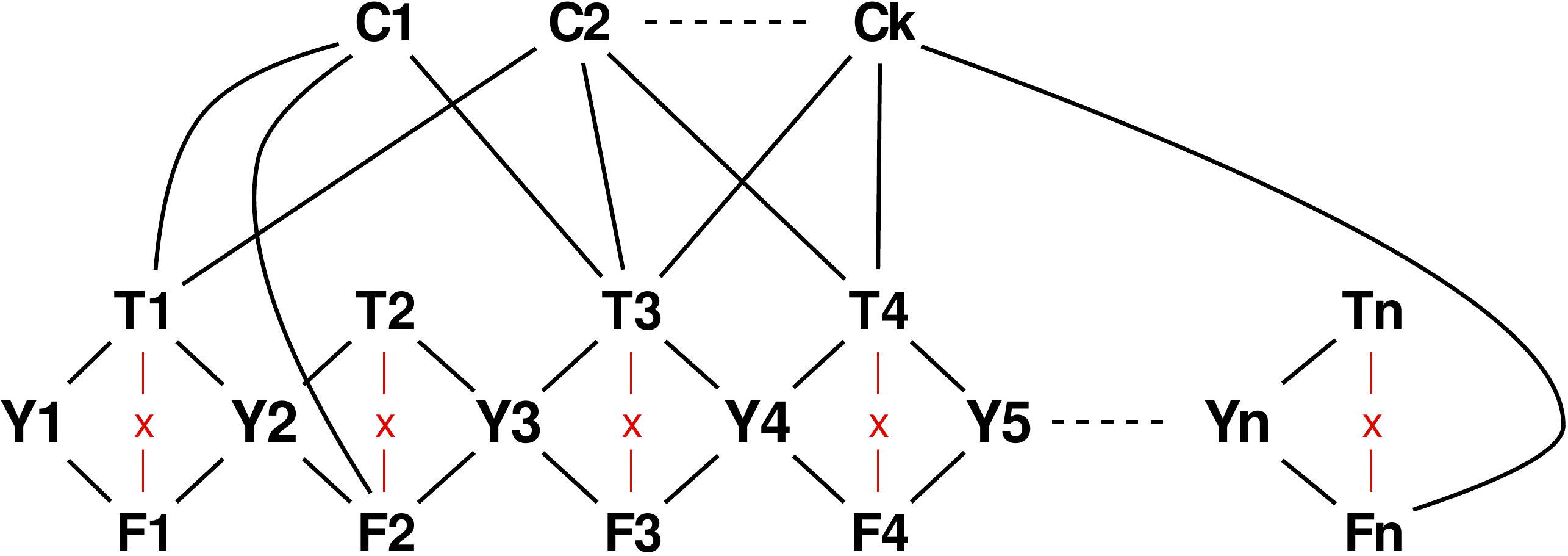}
    \caption{A graph with black and white edges. The crossed red edges between T1-F1, T2-F2 and so on symbolize that there is no white edge among those nodes, all other pairs of nodes (except the ones already connected by a black edge) are connected by a white edge. Computing the lexicographically minimum \bclique containing $X=\{Y1\}$ is NP-hard.}
    \label{fig:comp-hard}
\end{figure}

\begin{lemma}
\label{lem:hardness}
Given a graph $G$ whose edges are either black or white and a non-maximal \bclique $X$ of $G$, it is NP-hard to find the lexicographically minimum among the maximal \bcliques containing $X$.
\end{lemma}
\begin{proof}
We prove that a $\comp(X)$ function that returns the lexicographically minimum \bclique containing $X$ can be used to solve a SAT problem in polynomial time, by building a graph with nodes linear in the amount of the clauses and variables in the formula.

Given a SAT formula with $n$ variables $x_1 \ldots x_n$ and $k$ clauses $d_1 \ldots d_k$, we build the the graph in Figure~\ref{fig:comp-hard}, whose nodes are $C1 \ldots Ck$, $T1 \ldots Tn$, $F1 \ldots Fn$ and $Y1\ldots Yn$, labelled increasingly in this order (i.e., nodes $C1 \ldots Ck$ have smaller label than all other nodes). Each $Yi$ is connected with a \emph{black} edge to $Ti$ and $Fi$, and, except for $Y1$, also with $Ti-1$ and $Fi-1$. Each $Ci$, which corresponds to $d_i$, is connected with \emph{black} edge to $Tj$ (resp. $Fj$) whether $d_i$ contains a positive (resp. negative) occurrence of $x_j$. Hence, nodes in $C1\ldots Ck$ are connected with black edge to an arbitrary amount of $Ti$ and $Fi$ nodes, but not to any $Yi$ node. All other pairs of nodes are connected with a \emph{white} edge, except for the pairs $Ti$,$Fi$ (symbolized by the red crossed edge in Figure~\ref{fig:comp-hard}).
It is straightforward to see that any maximal \bclique in this graph will contain exactly one between any pair of nodes $Ti$,$Fi$, and that any maximal \bclique containing all nodes in $C1\ldots Ck$ will be lexicographically smaller than any that does not contain all of them (as they have the smallest labels).

Consider $\comp(\{Y1\})$, i.e., the lexicographically smallest maximal \bclique containing $Y1$. Any \bclique containing $Y1$ and all $Ci$ nodes represents a satisfying assignment for the formula at hand. Indeed, in order for each $Ci$ node to be reachable from $Y1$ with black edges, at least one of the $Tj$ or $Fj$ nodes connected to $Ci$ must be in the \bclique; the set of $Ti$ and $Fi$ nodes in the \bclique will thus give us the value (true or false) of the corresponding variable $x_i$ (recall that we cannot have the pair $Tj$-$Fj$ in the same \bclique).
Hence, in order to verify that the formula is satisfiable, we only need to compute $\comp(\{Y1\})$ and check whether this contains all $Ci$ nodes.
\end{proof}
Lemma~\ref{lem:hardness} represents an obstacle to implement reverse search algorithms using Property~\ref{prop:class-rev-search}.
Furthermore, a polynomial cost per solution cannot be guaranteed in the general case, as per the following lemma.

\begin{lemma}
\label{lemma:conditional_lb}
An algorithm that generates all the $\alpha$ maximal solutions of any strongly accessible set system has a worst case time complexity of $\Omega(\alpha 2^{q/2})$, unless SETH is false. 
\end{lemma}
\begin{proof}
Suppose by contradiction that an algorithm $A$ exists to list all the maximal solutions in an independence set system, which is a particular case of strongly accessible set systems, in $O(\alpha 2^{\frac{q}{k}})$ time for some $k > 2$. This implies the existence of an algorithm $B$ for SAT that runs in $O(n 2^{\frac{2n}{k}})$ time, where $n$ is the number of variables. Indeed, using the reduction in~\cite{lawler1980generating}, $B$ first transforms a formula of SAT in a set system, having $\alpha=n$ maximal solutions iff the formula is not satisfiable. Note that this reduction gives $q=2n$ for the size of the largest maximal set. At this point, $B$ can execute $A$ waiting $O(n\cdot 2^{\frac{2n}{k}})$ time: the formula is not satisfiable iff $A$ finds $n$ solutions and terminates within this upper bound. This contradicts the hypothesis that $k > 2$ as it must be $k \leq 2$ by SETH.
\end{proof}

\subsection{Using reverse search: canonical ordering}
\label{sec:our-framework}

We give a definition of \comp that is different from that in Property~\ref{prop:class-rev-search}, exploiting the definition of strongly accessible set systems: for any $X,Y\in \good$ such that $X\subset Y$, there exists $z\in Y\setminus X$ such that $X \cup\{z\}\in\good$. The idea is to let \comp incrementally add these elements $z$, thus \emph{permuting} the elements of a maximal solution $S$, and obtaining a relative order of the elements of $\univ$ that depends on $S$ itself.

Specifically, we label the elements in $\univ$ arbitrarily as $v_1, \ldots, v_{|\univ|}$, and implicitly enforce the condition $v_i < v_j$ iff $i < j$ after the labeling. From now on, we identify elements with their labels. Given $X\in\good$ and a subset of elements $A \subseteq \univ$, we use $X^+_A=\{a\in A\setminus X : X\cup \{a\}\in\good\}$ to denote the set of elements of $A$ that can be added to $X$. We use $X^+$ as an equivalent form for $X^+_{\univ}$. (Note that the minimum element of these sets is well defined because $\univ$ is labelled.) We use $Z=\{x \in \univ : \{x\}\in \good\}$ to denote the set of good singletons.\footnote{In the case of \bcliques, we have that every node is a \bclique (i.e., $Z=\univ$). However, in strongly accessible set systems, it can be $Z \subset \univ$: consider the set of simple paths that contain a leaf in a tree, where $\univ$ are the nodes of the tree; observe that $Z$ is the set of leaves.} 

The following fact holds.
\begin{fact}
\mbox{}\vspace*{-1.5ex}
\label{obs:zeta}
\begin{enumerate}
\item \label{item:i} $X \cap Z \neq \emptyset$ for any nonempty $X\in \good$.
\item \label{item:ii} $X^+_S \neq \emptyset$ for any $X,S\in\good$ such that $X\subset S$.
\end{enumerate}
\end{fact}
\begin{proof}
Fact~\ref{obs:zeta}.\ref{item:i} holds as $Y=\emptyset \subset X$ is in $\good$ implies that there exists $z\in X\setminus Y$ such that $Y\cup \{z\}=\{z\} \in \good$ by definition of strongly accessible set systems. Hence, $z\in Z\cap X$ meaning that $Z\neq\emptyset$ if a nonempty solution exists. Fact~\ref{obs:zeta}.\ref{item:ii} directly follows from the definition of strongly accessible set system.
\end{proof}

\begin{algorithm}[t]
\caption{Basic framework for strongly accessible set systems}
\label{alg:main}
\small
\begin{minipage}[t]{0.38\textwidth}
\SetKwInOut{Input}{Input}
\SetKwInOut{Output}{Output}
\nonl
\Input{Strongly accessible $(\univ,\good)$}
\Output{All maximal $X\in\good$}
\DontPrintSemicolon
\LinesNotNumbered
\BlankLine
\nonl 
\nonl\For{$R\in\roots$}{
    \nonl $\rec( R )$
}
\BlankLine
\SetKwProg{myproc}{Function}{}{}
  \myproc{$\comp(X,A)$}{
    \While{$X^+_A \neq \emptyset$}{
    	$x\gets \choice(X,A)$\;
        $X\gets X\cup \{x\}$\;
    }
    \Return{X}   
  }
\BlankLine
\SetKwProg{myproc}{Function}{}{}
  \myproc{$\choice(X,A)$}{
  \Return{$\min X^+_A$}   
  }
\end{minipage}~\begin{minipage}[t]{0.62\textwidth}
\SetKw{Yield}{yield return}
\DontPrintSemicolon
\LinesNotNumbered
\SetKwProg{myproc}{Function}{}{}
  \nonl \myproc{$\rec(X)$}{
  \textbackslash* Output $X$ if depth is odd *\textbackslash\;
  \ForEach{$w\in \univ$}{
    \ForEach{$S\in \children(X,w)$}{
    		\rec($S$)
    }
  }
  \textbackslash* Output $X$ if depth is even *\textbackslash
}
\BlankLine
\myproc{$\children(P, w)$}{
        \ForEach{$X\subset P\setminus\{w\}$ such that $X \cup \{w\} \in \good$}{
            $S\gets \comp(X\cup \{w\}, \univ)$\;
            
            \If{$\langle \parent(S), \parind(S), \core(S)\rangle = \langle P,w,X\rangle$}{
                \Yield{$S$}
            }
        }
    }
\end{minipage}
\end{algorithm}

Given $X \in \good$, we denote by $\source(X)$ the smallest element in $X\cap Z$, which exists by Fact~\ref{obs:zeta}. Given a set $A\subseteq \univ$, we define $\comp(X,A)$ as a procedure that uses Fact~\ref{obs:zeta}: while $X^+_A$ is nonempty, in each iteration an element $x \in X^+_A$ is selected by a user-defined black box $\choice(X,A)$, and $x$ is added to $X$ (see Algorithm~\ref{alg:main} where $\choice(X,A)$ returns the minimum). Different implementations of \choice lead to different results for our framework. We use $\comp(X)$ as a shorthand for $\comp(X,\univ)$.

Let us consider $\comp(\{\source(S)\},S)$: it naturally induces a permutation of the elements in $S$, called \emph{canonical order} $s_1,\ldots, s_{|S|}\in S$, where $\source(S)$ is the first element $s_1$, and each subsequent element $s_i$ is chosen during the iterations ($i=2,3, \ldots, |S|$) of $\comp$ according to $\choice$. This differs from Section~\ref{sec:classical-reverse-search}, where the increasing order of $S$ is implicitly adopted.

\begin{example}
In the case of the \bcliques in Figure~\ref{fig:bccliques}, let $X$ be $\{1,2\}$, $X^+$ is the set $\{5,6\}$ as both the nodes $5$ and $6$ can be added to $X$ while maintaining the \bclique property ($3$ is not in $X^+$ as it is not black connected to $X$). $\comp(X)$ adds to $X$ the node $5$, since $5<6$; it then adds $3$, as $X^+=\{3,6\}$, and finally $6$, as $X^+=\{6\}$. When $X=\{1,2,3,5,6\}$, $X^+=\emptyset$, as $X$ cannot be enlarged anymore maintaining the \bclique property. A visual representation of the corresponding canonical order is shown top-down in Figure~\ref{fig:clq-forms}(b). 
\qexd
\end{example}

Note that $\comp(X)$ may not return the lexicographically smallest \bclique containing $X$, but we will show that its properties allow us to define a suitable \parent function.

Given $S$ and $0 < j\leq |S|$, we define $S[j]= s_1,\ldots, s_{j}$ as the $j$th prefix of the canonical order of $S$. We observe that the following prefix-closure property holds for any \choice.

\begin{lemma}
\label{prop:pref-good}
For any maximal solution $S\in \good$, we have $S[j]\in \good$ for any $1\leq j\leq |S|$.
\end{lemma}
\begin{proof}
We proceed by induction. $S[1]=\{s_1\}=\{\source(S)\}\in \good$, by definition of $Z$ as $s_1\in S\cap Z$. Assuming that $S[j]\in \good$, we have that also $S[j+1]$, which is obtained from $S[j]\cup \{s_{j+1}\}$, is in $\good$ by definition of $S[j]^+_S$, as $s_{j+1}\in S[j]^+_S$. 
\end{proof}

Along with \choice, our framework also requires a user-defined total ordering\footnote{We say that $S\prec T$ for any two maximal solutions $S$ and $T$ when $S\preceq T$ and $S\neq T$.} $\preceq$ between maximal solutions that satisfies the following constraint, needed to satisfy Definition~\ref{prop:rev-search}.\ref{rev-b}.

\begin{requirement}
\label{prop:pref-min}
For any maximal solution $S\in \good$, $\comp(S[j])\preceq S$ for any $j\leq |S|$.
\end{requirement}

\begin{lemma}
\label{lem:comp-strongly}
Algorithm~\ref{alg:main} fulfills Requirement~\ref{prop:pref-min}, where $\preceq$ is defined by lexicographically comparing the canonical orders of maximal solutions.\footnote{Given any two maximal solutions $S \neq T$, consider their canonical orders $s_1,\ldots,s_{|S|}$ and $t_1,\ldots,t_{|T|}$. Let $i$ be the smallest index for which $s_i \not= t_i$, which exists by their maximality. Then, $S\prec T$ iff $s_i < t_i$.}
\end{lemma}
\begin{proof}
If $\comp(S[j]) = S$, we have nothing to prove. Otherwise, let the canonical order of $\comp(S[j])$ and $S$ be respectively $t_1,\ldots,t_{|T|}$ and  $s_1,\ldots,s_{|S|}$ and let $i > j$ be the smallest index such that $t_i\neq s_i$. We have $S[i-1]=T[i-1]$, so $t_i$ is the minimum of $S[i-1]^+$ and, since $s_i\in S[i-1]^+$, we have $t_i\leq s_i$.
\end{proof}

\begin{example}
\label{ex:can-ord}
In the case of \bcliques, refer to Figure~\ref{fig:bccliques}: we can see that the \bcliques $S_1=\{1,2,3,5,6\}$ and $S_2=\{3,4,5\}$ have respectively $\langle 1,2,5,3,6 \rangle$ and $\langle 3,5,4 \rangle$ as canonical orders. By comparing the two sequences we see that $\langle 1,2,5,3,6 \rangle \prec \langle 3,5,4 \rangle$ as $1<3$.
\qexd
\end{example}

Continuing the description of our framework, we can now define the notions required for the reverse search (Definition~\ref{prop:rev-search}).

\begin{definition}
\label{def:parent}
Given a maximal solution $S$, define $\parind(S)$ as the earliest element $s_j$ in the canonical order of $S$ such that $\comp(S[j])=S$: if $j>1$, then $\core(S) = S[j-1]$, and $\parent(S) = \comp(\core(S))$; otherwise, $S$ has no parent and thus $S\in\roots$.
\end{definition}

\begin{example}
Consider the \bcliques $S_1=\{1,2,3,5,6\}$ and $S_2=\{3,4,5\}$ in Figure~\ref{fig:bccliques}. Their canonical order is respectively $\langle 1,2,5,3,6 \rangle$ and $\langle 3,5,4 \rangle$. Since $\comp(\{1\})=S_1$, we have $S_1\in\roots$. For $S_2$, since both $\comp(\{3\})$ and $\comp(\{3,5\})$ give $S_1$, $\parind(S_2)=4$ and $\core(S_2)=\{3,5\}$. As $\comp(\core(S_2))=S_1$ we have that the parent of $S_2$ is $S_1$. It is worth observing that $S_1=\parent(S_2)\prec S_2$ (see Example~\ref{ex:can-ord}).
\qexd
\end{example}

By Requirement~\ref{prop:pref-min}, we have $\parent(S)\prec S$ as $\parent(S) = \comp(\core(S)) \neq S$, and thus it satisfies Definition~\ref{prop:rev-search}.\ref{rev-b}. 
From Definition~\ref{def:parent}, we have that $S$ is a root iff  $\comp(S[1])=S$, thus satisfying Definition~\ref{prop:rev-search}.\ref{rev-c}.
As for Definition~\ref{prop:rev-search}.\ref{rev-d}, we observe that $\children(S)$ is defined as all the solutions $S'$ for which $\parent(S')=S$. 

\begin{theorem}
\label{the:main-strongly}
Given a total order $\preceq$ among maximal solutions and a $\choice$ function such that Requirement~\ref{prop:pref-min} is fulfilled, and given $\children(S)$, any strongly accessible set system is reverse searchable.
\label{lem:all}
\end{theorem}

We now discuss how to design $\children(S)$. Let $\children(S,w)$ be the function that returns all the maximal solutions $S'\in\children(S)$ such that $w=\parind(S')$. Clearly, $\children(S)=\bigcup_{w\in \univ} \children(S,w)$, and the union is disjoint since a solution has a unique $\parind$ (Definition~\ref{def:parent}).

With reference to Algorithm~\ref{alg:main}, we show how to implement $\children(P,w)$ given a maximal solution $P$ (the parent) and one element $w\in \univ$. 

\begin{lemma}
\label{lem:strongly-children}
For any two maximal solutions $S$ and $P$, such that $P=\parent(S)$, there exists $X\subset P$ and $w\not\in X$ such that $S=\comp(X\cup \{w\})$, where $X=\core(S)$ and $w=\parind(S)$.
\end{lemma}
\begin{proof}
According to Definition~\ref{def:parent}, given a solution $S$, whose canonical order is $s_1,\ldots,s_{|S|}$, the relationship between $S$ and $P=\parent(S)$ is such that $P=\comp(\core(S))$, where $\core(S)=s_1,\ldots, s_j$ for some $j$. For each $i$, with $1\leq i\leq j$, $s_i\in S\cap P$; $s_{j+1}=\parind(S)$ is in $S\setminus \core(S)$.
Note that given $S$ and $P=\parent(S)$, we have $\core(S)\subset P$, $\parind(S)\not\in \core(S)$, and $S=\comp(\core(S)\cup \{\parind(S)\})$. 
\end{proof}

Lemma~\ref{lem:strongly-children}, essentially says that we can exploit algorithmically the inclusion $\children(P,w)\subseteq \{\comp(X\cup\{w\}) : X\subset P ,\ X\cup\{w\}\in\good \}$.
It is worth observing that while $\core(S)$ is a prefix of $S$, it is not necessarily also a prefix of $P$ as the canonical order of $P$ could differ (including its $\source$). This implies that we do not know a priori which subsets $X \subset P$ may be the \core of some child of $P$. The only hint we have is that $X\cup\{w\}\in\good$ as its elements must be a prefix of $S$ when put into canonical order.\footnote{This is the case, for instance of the \bcliques in Figure~\ref{fig:bccliques}: we can see that the parent of \bclique $\langle 2,5,8,7 \rangle$ is the \bclique $\langle 1,2,5,3,6 \rangle$, its core is $\langle 2,5 \rangle$ and its parent index is the node $8$.}

$\children(P,w)$ computes candidate children $S$ as $\comp(X\cup\{w\})$, for each possible $X\in 2^P$ such that $X\cup\{w\}\in\good$: it retains $S$ if $P=\parent(S)$, $w=\parind(S)$, and $X=\core(S)$. While the first two conditions are sufficient to ensure that $S\in \children(P,w)$, the third one is required to output each child once: if we find an $S$ such that $P=\parent(S)$, $w=\parind(S)$ but $X\neq \core(S)$, we discard it as surely we will find $S$ again when $X = \core(S)$. 

\begin{example}
Consider the case in Figure~\ref{fig:bccliques}. For the sake of simplicity we will report the elements of a solution in the canonical order. When $P=\{3,5,4\}$, $X = \{5\}$ and $w = 8$, the candidate child $S$ in Algorithm~\ref{alg:main} is $\comp(\{ 5,8 \} ) = \{2,5,8,7\}$. However, $\parent(2,5,8,7) = \{1,2,5,3,6\}$, thus we discard $S$. Now consider $P=\{1,2,5,3,6\}$, $X=\{2\}$ and $w=8$: the corresponding candidate child $S$ is $\{2,5,8,7\}$. Even though $\parent(2,5,8,7) = \{1,2,5,3,6\}$, we have  $\core(\{2,5,8,7\})=\{2,5\}\neq X$ and hence we discard $S$ again. When $P=\{1,2,5,3,6\}$, $X=\{2,5\}$, and $w=8$, we obtain the same $S$ a third time, but now $S$ is not discarded.
\qexd
\end{example}

The running time of Algorithm~\ref{alg:main} is stated in the following lemma. It is worth noting that the delay is close to the conditional lower bound of Lemma~\ref{lemma:conditional_lb}. Recall that $\goodtime$ is the time needed to check whether $X\in \good$.
\begin{lemma}
\label{lem:time}
Algorithm~\ref{alg:main} has delay $O(q2^q\cdot |\univ|^2 \cdot \goodtime)$.
\end{lemma}
\begin{proof}
$\choice$ can clearly be computed in $O(|\univ| \cdot \goodtime)$ time. This means that we can compute $\comp$ in $O(q|\univ| \cdot \goodtime)$. We can compute $\parent$, $\parind$ and $\core$ in $O(q|\univ| \cdot \goodtime)$ too, since it is enough to check for every $j$ if $\choice(S[j]) = s_j$, as the greatest $j$ for which this does not happen identifies the $\core$ and the $\parind$. To compute $\parent$, we then only need to run $\comp$ once. Since we run these computations $2^q\cdot|\univ|$ times for every solution found, the thesis follows.
\end{proof}

\section{Refined Framework for Connected Hereditary Properties}
\label{sec:commutable}

A natural question is whether a more refined method to generate children is possible. As we show in this section, this has a positive answer at least for connected hereditary graph properties, due to what we call the \oursys property: 

\begin{definition}[Commutable property]
A strongly accessible set system $(\univ,\good)$ verifies the \emph{\oursys} commutable property if, given any $X,Y \in \good$ such that $X\neq \emptyset$ and $X \subset Y$, for any $a,b\in Y\setminus X$, if $X\cup\{a\} \in\good$ and $X \cup\{b\} \in\good$ then $X\cup\{a,b\}\in\good$. 
\end{definition}

We will show that connected hereditary properties (that is, the set systems induced by connected hereditary properties) also verify this \oursys property, and that this will allow us to design a smarter way to generate children solutions.

\begin{lemma}
The set systems induced by connected hereditary graph properties are strong accessible and verify the \oursys property.
\end{lemma}
\begin{proof}
Let $(\univ,\good)$ be the set system induced by a connected hereditary property, and let $X,Y \in \good$ be two solutions such that $X \subset Y$.
Firstly, there must be $v\in Y\setminus X$ connected to some vertex in $X$, otherwise $Y$ would not be connected, as no vertex in $X$ would be connected to any vertex in $Y\setminus X$; since $Y\in\good$ and thus $Y$ is connected, we have some $v\in Y\setminus X$ connected to some vertex in $X$, and $X\cup\{v\}\in\good$ by definition of connected hereditary, thus $(\univ,\good)$ is strongly accessible.
Furthermore, assume $X\ne \emptyset$, and let $a,b\in Y\setminus X$ be two elements such that $X\cup\{a\} \in\good$ and $X\cup \{b\}\in \good$ .
We thus have that $X$ is connected, and so are $X\cup\{a\}$ and $X\cup\{b\}$. Since $X$ is not empty it must be that $X\cup\{a,b\}$ is connected. We have that $X\cup\{a,b\}$ is a connected subgraph of $Y$, which by definition of connected hereditary property implies $X\cup\{a,b\}\in\good$, thus $(\univ,\good)$ verifies the \oursys property.
\end{proof}

Again, we define $\choice$ (and thus $\preceq$) so that $\comp$ creates a canonical order.
To do so we first need to introduce the notion of \level.

\begin{definition}
\label{def:layer}
Given $X\in\good$ and a \emph{starting element} $t\in X \cap Z$, define inductively $B_i$:\footnote{Note that $B_i$ is made only of elements from $X$ whereas $B_i^+$ is made of elements from $X\cup X^+$.}
\begin{itemize}
\item $B_0=\{t\}$
\item $B_{i}=B_{i-1}\cup (B_{i-1}^+\cap X)$, for $i>0$
\end{itemize}
Then, for any $y\in X\cup X^+$, its \emph{\level} relatively to $X$ from $t$ is $\lev_t^X(y)=\min\{i:y\in B_{i-1}^+\}$.
\end{definition}

We simply write $\lev^X$ when $t=\source(X)$. This is useful for defining $\choice(X,A)$, as given in Algorithm~\ref{alg:ristr}: for each element $y \in X^+_A$, associate a pair $\langle \lev^X(y), y \rangle$ with $y$; return the element $y \in X^+_A$ whose associated pair is the lexicographically minimum. When $\choice$ is plugged into $\comp$, as we saw in Section~\ref{sec:our-framework}, we obtain the canonical order of a maximal solution $S$ and, consequently, the notions of \parent, \core, \parind.

\begin{figure}[t]
    \centering
    \includegraphics[width=.8\textwidth]{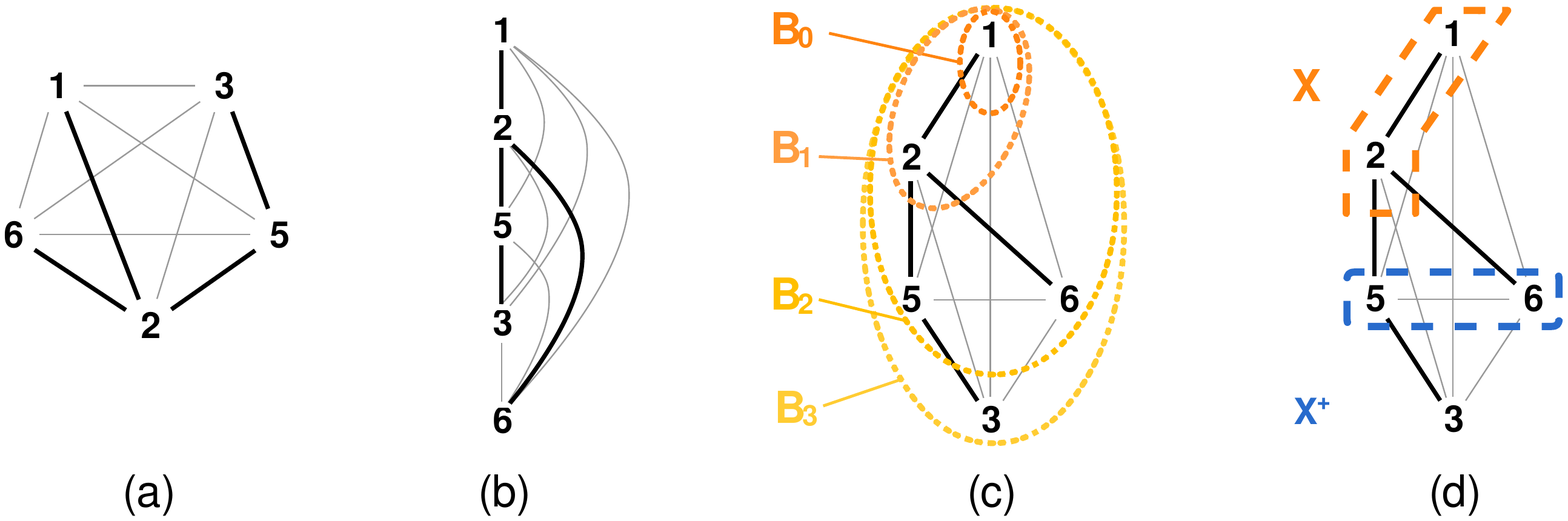}
    \caption{(a) A maximal \bclique, (b) its canonical order $1,2,5,3,6$ from Section~\ref{sec:our-framework}, and (c) its canonical order $1,2,5,6,3$ from Section~\ref{sec:commutable}. (d) A non-maximal \bclique $X$ and its set $X^+$.} 
    \label{fig:clq-forms}
\end{figure}

\begin{algorithm}[t]
\caption{Refined framework for connected hereditary properties: $\children + \choice$}
\label{alg:ristr}
\small
\DontPrintSemicolon
\SetKw{Yield}{yield return}
\BlankLine\nonl 
\begin{minipage}[t]{\textwidth}

\SetKwProg{myproc}{Function}{}{}
  \myproc{$\children (P, w)$}{
        \lIf{$w\in P$}{\Return{$\emptyset$}}
        \ForEach{$R\in \restr(P,w) \setminus \{P\}$}{
            \ForEach{$s \in (R\cap Z) \setminus \{w\}$}{
                $\mathit{coreS} \gets \comp(\{s\},R)|_w$\;
                $S \gets \comp(\mathit{coreS} \cup \{w\}, \univ)$\;
                
                \If{$\langle \parent(S), \parind(S), \erre(S), \source(S)\rangle = \langle P,w,R,s\rangle$}{
                    \Yield{$S$}
                }
            
            }
        }
    }
\end{minipage}~\hspace*{-0.38\textwidth}\begin{minipage}[t]{0.35\textwidth}

\SetKwProg{myproc}{Function}{}{}
  \myproc{$\choice (X, A)$}{
  \Return{$\argmin\limits_{y\in X^+_A}\, \langle \lev^X(y), y \rangle$}   
  }
 
\end{minipage}
\end{algorithm}

\begin{example}
Consider the \bclique in Figure~\ref{fig:clq-forms}(a), where $X = \{1,2,3,5,6\}$ is the set of its vertices, and the starting element is $t = \source(X) = 1$. The elements of $X$ have layers $0, 1, 3, 2, 2$, respectively, relatively to $X$ from $t$ (e.g. $\lev^X(2)=1$ and $\lev^X(5)=2$). Running $\comp( \{\source(X)\}, X)$ with the \choice function in Algorithm~\ref{alg:ristr}, we obtain the canonical order in Figure~\ref{fig:clq-forms}(c) as follows. We add~$2$ to $\{1\}$ as it is the only option for \choice. Next, we can choose between~$5$ and~$6$ for extending~$\{1,2\}$: we adds~$5$, as $\langle\lev^X(5),5\rangle$ is smaller than $\langle\lev^X(6),6\rangle$. This gives $\{1,2,5\}$ and \choice selects $6$ between the next options $6$ and $3$. Finally, it selects $3$, obtaining $X$ and the above canonical order. Also, we observe that $X \in \roots$ as \comp applied to each prefix $S[j]$ (i.e. $1$; $1,2$;~\dots; $1,2,5,6,3$) always gives $X$.

Consider now the other two \bcliques, $X_1 = \{3,4,5\}$ and $X_2 = \{2,5,7,8\}$, from Figure~\ref{fig:bccliques}. Their canonical order is $3|5|4$ (with layers in increasing order separated by ``$|$'') and $2|5,8|7$. Also, $X$ is their parent: for example, $X = \parent(X_1) = \comp(\{3,5\})$ and thus $\core(X_1) = \{3,5\}$ and $\parind(X_1) = 4$. We observe that during the execution of \comp, the \source changes (hence the canonical order): $\{3|5\} \rightarrow \{2|5|3\} \rightarrow \{1|2|5|3\} \rightarrow  \{1|2|5,6|3\}$. 
(Also, $\core(X_2) = \{ 2,5\}$ and $\parind(X_2) = 8$.)
In general, during the execution of $\comp$ the \source of $X$ may change and, consequently, the \levels of the elements in $X\cup X^+$ change accordingly. This does not affect $X^+$, but only the relative order of its elements. 
\qexd
\end{example}

Recalling that we have to meet Requirement~\ref{prop:pref-min}, we define the order $\preceq$ between any two maximal solutions $S$ and $T$ as follows. Let their canonical order be $s_1,\ldots,s_{|S|}$ and $t_1,\ldots,t_{|T|}$.
Then $S\prec T$ iff $\langle \lev_{s_1}^S(s_j),s_j \rangle$ is lexicographically smaller than $\langle \lev_{t_1}^T(t_j),t_j \rangle$, for the smallest $j$ giving different pairs (note that $j$ exists as both solutions are maximal).

We are now ready to state the version of Lemma~\ref{lem:comp-strongly} for commutable set systems, which can be proven in the same way as the version for strongly accessible set systems.

\begin{lemma}
\label{lemma:requirement-alg-ristr}
When combined with \choice from Algorithm~\ref{alg:ristr}, \comp in Algorithm~\ref{alg:main} fulfills Requirement~\ref{prop:pref-min}.
\end{lemma}

We finally need to implement \children to meet the conditions of Theorem~\ref{the:main-strongly} for reverse search. A key ingredient is the concept of \emph{restricted problem}, originally introduced by Lawler et al.~\cite{lawler1980generating}:
given a maximal solution $P\in\good$ and an element $w\in \univ\setminus P$, the restricted problem $\restr(P,w)$ asks to list all maximal solutions $R \neq P$ in the reduced set system $(P\cup\{w\}, \good)$. Under suitable conditions, Cohen et al.~\cite{cohen2008generating} remarkably prove that we can efficiently enumerate the maximal solutions in (connected) hereditary systems $(\univ, \good)$ iff we can do that in its restricted version for $(P\cup\{w\}, \good)$.
We will assume that we can enumerate the solutions of any restricted problem $\restr(P, w)$ in at most $\restrtime$ time, and using at most $\restrspace \ge q$ space. Moreover, we will use $\restrbound$ to denote an upper bound on the number of solutions of any restricted problem.
Using $\restr(P,w)$ as a black box we give some necessary conditions to compute $\children(P,w)$, recalling that $\children(P)=\bigcup_{w\in \univ} \children(P,w)$. 

The main conceptual step is showing that examining the solutions of $\restr(P,w)$ is enough for our goal. Define the truncated function $\comp(X,A)|_w$ as follows: during the execution of \comp, if $\choice(X,A)$ returns $w$ then stop the execution and return the current $X$  (so that $w \not \in X$). 
We need the above truncated function and this key lemma to implement $\children(P,w)$ in Algorithm~\ref{alg:ristr}.

\begin{lemma}
\label{lemma:key}
For any maximal solution $P$ and $w \in \univ$, we have that $\children(P,w)$ is contained in the set of maximal solutions $S$ such that
\begin{enumerate}
\item \label{item:key:1} there exists $R\in \restr(P,w)$ (with $R \neq P$ and $w \not \in P$) and
\item \label{item:key:2} $S = \comp(\, \comp(\{s\},R)|_w \cup \{w\}, \, \univ\,)$ (with $s \in (R\cap Z)$ and $s \neq w$).
\end{enumerate}
\end{lemma}
\begin{proof}
We will prove that, for any $S$ such that $\parent(S) = P$, the two conditions of the lemma hold by choosing $w = \parind(S)$, $s = \source(S)$ and $R = \erre(S) = \comp(\core(S) \cup \{\parind(S)\}, \parent(S) \cup \{\parind(S)\})$\footnote{Note that $\erre(S)$ is well defined as $\core(S)\cup \{\parind(S)\}\in\good$, by Lemma~\ref{prop:pref-good}, as it is a prefix of $S$.}.

During the proof, we will use repeatedly the following \emph{key fact} about our $\choice$ function: \textit{for any $Y \subseteq A$, if $b = \choice(X, A)$ belongs to $(X\cup Y)^+_A$, then $b = \choice(X \cup Y, A)$} (provided that the starting element stays the same). Indeed, anything in $Y\setminus X$ must have a \level at least as large as $b$, so anything in $(X \cup Y)^+_A$ that is not in $X^+_A$ must have a larger layer than $b$. Note that this property does not hold if we use the order of Section~\ref{sec:our-framework}.

By construction of \comp, $\erre(S)$ is a maximal solution in $\parent(S) \cup \{\parind(S)\}$ and it is different from $\parent(S)$, as $\parind(S) \in \erre(S)$, provided that we can prove $\parind(S) \not \in \parent(S)$. Suppose by contradiction that $w = \parind(S) \in \parent(S)$, and let $y = \choice(\core(S))$ be the first element that $\comp$ adds to $\core(S)$ when generating $\parent(S)$. Note that $y \neq w$ by definition of \core. Clearly $y \not \in S$, as otherwise $\core(S) \cup \{y\}$ would be a prefix of $S$ and so $y = \parind(S) = w$. By the \oursys property, the fact that $\core(S) \cup \{w\} \cup \{y\} \subseteq \parent(S)$ and both $\core(S) \cup \{w\}, \core(S) \cup \{y\} \in \good$ implies that $y \in (\core(S) \cup \{w\})^+$. By the aforementioned key fact we then have that $y = \choice(\core(S) \cup \{w\})$, so $\comp(\core(S) \cup \{w\}) \neq S$, which contradicts the definition of $\core(S)$ and $\parind(S)$. This proves the first part of the statement.

Regarding the second part, it is enough to prove that $\comp(\{s\}, R)|_w = \core(S)$. We will only consider the case in which $\core(S) \cup \{w\} \neq S$, as otherwise the proof is trivial (since $\erre(S) = S$). Let $s = s_1, \dots, s_{cs}$ be the canonical order of $\core(S)$, where $cs = |\core(S)|$ and $w$ is denoted as $s_{cs+1}$. Suppose now that the first $cs$ elements added by $\comp(\{s\}, R)$ are not $s_2, \dots, s_{cs+1}$ and let $r_i \neq s_i$, $i>1$, be the first ``wrong'' element that is added. Thus we have $r_i = \choice(S[i-1], R)$, recalling that $S[i-1] = \{s_1, \dots, s_{i-1}\}$. Note that $r_i \not \in S$, as that would contradict the definition of canonical order. By the \oursys property and the fact that and $\core(S) \cup \{w\} \cup \{r_i\} \subseteq R$, we can easily prove (by induction) that $r_i \in S[j]^+_R$ for $j=i-1, \dots cs+1$. Since $S[cs+1] = \core(S)\cup\{w\}$, we have $r_i \in (\core(S) \cup \{w\})^+_R$ and, by the aforementioned key fact, $r_i = \choice(\core(S) \cup \{w\}, R)$. This implies that $r_i$ is a better candidate than $s_{cs+2}$ for $\choice$: this continues to be true if we consider the whole $\univ$ instead of just $R$, so the first step of $\comp(\core(S) \cup \{w\})$ does not add $s_{cs+2}$. As a result, we have $\comp(\core(S) \cup \{w\}) \neq S$, which is a contradiction.
\end{proof}

We remark that Lemma~\ref{lemma:key} does not necessarily hold if we use the canonical order from Section~\ref{sec:our-framework} (as briefly noted in the proof), and avoids us to blindly target all of $X \subset P \cup \{w\}$.

We observe that different solutions $R \in  \restr(P,w)$ could lead to the same maximal solution $S$. We avoid this by checking that $R = \erre(S)$ when we generate a child, as $\erre(S)$ is computed deterministically from $S$. 

The following theorem holds.
\begin{theorem}
\label{the:commutable-enumeration}
Given a solution $P$, an element $w\in \univ\setminus P$, a solution $R\neq P$ of the restricted problem $\restr(P,w)$, there exists an algorithm that enumerates all maximal solutions $S$ such that $\parent(S) = P$, $\parind(S) = w$, $\erre(S) = R$ in time $O(q^3|\univ|\goodtime)$ and space $O(\goodspace)$.
\end{theorem}
\begin{proof}
The proof follows from Lemma~\ref{lemma:key}. Indeed, note that $\source(S) \in \core(S) \subset \erre(S)$. As the tuple $\langle \parent(S), \parind(S), \start(S), \erre(S)\rangle$ uniquely identifies $S$, and there are at most $|R\setminus \{w\}| < q$ possible sources, there are at most $q$ possible solutions to test. As each can be found by using Lemma~\ref{lemma:key} in time $O(q^2|\univ|)$,\footnote{A naive implementation of $\parent$ takes $q$ runs of $\comp$, but it can actually be computed by a single run of $\comp$ in which at every step we check if the next node to be added would be the next node in $S$ or not.} the statement holds. At most $O(q)$ space is used by the algorithm, giving a total space usage of $O(q+\goodspace) = O(\goodspace)$.
\end{proof}

As we scan all the possible solutions of the restricted problem, by Definition~\ref{prop:rev-search}, we obtain that Algorithm~\ref{alg:main}, using \children\ and \choice\ routines in Algorithm~\ref{alg:ristr}, makes any set system induced by a connected hereditary property reverse searchable.

\begin{theorem}
\label{the:commutable-running}
The delay of Algorithm~\ref{alg:main} using routines in Algorithm~\ref{alg:ristr} is $O(\restrtime \cdot |\univ| + q^3|\univ|\cdot\goodtime\cdot\restrbound \cdot|\univ|)$.
\end{theorem}
\begin{proof}
For every solution, we solve $|\univ|$ times the restricted problem. Then, for every solution of the restricted problem, we run the computations described in Lemma~\ref{the:commutable-enumeration}. The total running time follows easily.
\end{proof}

\section{Stateless Enumeration and Memory Usage}
\label{sec:stateless}
The recursive version of our algorithms do not yet guarantee polynomial space: indeed, the number of nested recursive calls could be non-polynomial. Moreover, Algorithm~\ref{alg:ristr} stores the solutions of the restricted problem, which may be non-polynomial in number.

We address the first of these problems by removing the explicit recursion. The state of the computation inside a certain recursive call is fully determined by the variables $P$, $w$, $X$ in Algorithm~\ref{alg:main} and $P$, $w$, $R$, $s$ in Algorithm~\ref{alg:ristr}. Moreover, when a recursive call is made in the algorithms the conditions written in the code imply that we can easily (and cheaply) compute the state variables using only information about the child. It is thus easy to modify these two algorithms to simulate the recursion avoiding an explicit stack.

With regard to the restricted problem, note that we can iterate over the solutions of the restricted problem using $\restrtime$ time and $\restrspace$ space: we can restart the iteration whenever we backtrack in the (simulated) recursion tree, as this does not impact the delay. These observations allow us to state the following.

\begin{theorem}
\label{the:space-bounds}
The stateless versions of Algorithm~\ref{alg:main} and Algorithm~\ref{alg:ristr} achieve the same delays while taking, respectively, $O(\goodspace)$ and $O(\goodspace+\restrspace)$ space.
\end{theorem}

Note that in the case of \bcliques, $\goodtime$, $\goodspace$, $\restrtime$ and $\restrspace$ are $O(q)$. 
Hence, as an application of Theorems~\ref{the:commutable-running} and~\ref{the:space-bounds}, we obtain the following.
\begin{theorem}
There exists an algorithm to list all the maximal \bcliques with delay $O(q^5|\univ|^2)$ and $O(q)$ space.
\end{theorem}

\begin{corollary}
\label{the:commutable-restricted}
For a \oursys set system, the solutions of the general problem can be enumerated in polynomial total time and polynomial space if and only if the solutions of the restricted problem can be enumerated in polynomial total time and polynomial space. Moreover, the solutions of the general problem can be enumerated in polynomial delay and polynomial space if the solutions of the restricted problem can be enumerated in polynomial time and polynomial space.
\end{corollary}


\newcommand{\node}{\textsc{node}\xspace}
\newcommand{\iterrec}{\textsc{improved-}\rec}
\newcommand{\nextnode}{\textsc{next-}\node}
\newcommand{\nextrestr}{\textsc{next-}\erre}
\newcommand{\nextchild}{\textsc{next-child}\xspace}
\newcommand{\parentstate}{\textsc{parent-state}\xspace}
\newcommand{\isroot}{\textsc{is-root}\xspace}
\newcommand{\certificate}{\textsc{child-exists}\xspace}
\newcommand{\nll}{\rm{\texttt{null}}\xspace}


\section{Conclusions}

In this paper we described a space-efficient framework to design algorithms for enumerating the maximal solutions in strongly accessible set systems, whose complexity is $O(\alpha 2^q)$ time. This is, to the best of our knowledge, the first non-trivial bound for this class of enumeration problems, that approaches the lower bound of $O(\alpha 2^{q/2})$ time given by Lawler~\cite{lawler1980generating}. 

Furthermore, we solve the open problem left by Cohen et al.~\cite{cohen2008generating} by giving an improved version of our framework which links its delay to the \textit{restricted problem}, but still using polynomial space, for connected hereditary graph properties.

To give a more complete picture, we should remark that the improved version of our framework is not limited to just connected hereditary graph properties: we call the class of set systems to which it can be applied \oursys set systems, which corresponds to all set systems that are strongly accessible and respect the \oursys property.
It is easy to show that the class of \oursys set systems includes more than connected hereditary properties: for example, it is easy to see that required hereditary and required connected hereditary properties (defined in Section~\ref{sec:introduction}) verify the \oursys property, thus both version of our framework can be applied to them.

Future work is aimed at applying this framework to a variety of problems, with some already promising preliminary results. A problem that remains open is applying reverse search to a wider context of problems, e.g., without relying on the \oursys property or strong accessibility.

\newpage

\bibliographystyle{abbrv}

\newpage
\appendix

\section*{APPENDIX}

\section{Stateless Algorithm}

\begin{algorithm}[h]
\small
\SetKwInOut{Input}{Input}
\SetKwInOut{Output}{Output}
\SetKwRepeat{DoWhile}{do}{while}
\caption{Stateless framework with minimal memory}
\label{alg:iterdesc}
\BlankLine
\DontPrintSemicolon
\SetKwProg{myproc}{Function}{}{}
\myproc{\iterrec (X)}{
    $P\gets X$\;
    $S\gets \nll$\;
    $w\gets \nextnode(\nll)$\;
    $R\gets \nextrestr(P,w,\nll)$\;
    \DoWhile{true}{
        \DoWhile{$w\gets\nextnode(w)\neq \nll$\label{ln:ncand}}{
            \DoWhile{$R\gets\nextrestr(P,w,R)\neq \nll$\label{ln:nrestr}}{
                \If{$S \gets \nextchild(P,w,R,S) \neq \nll$\label{ln:nchild}}                {
                 $\langle P,S,w,R \rangle \gets \langle S, \nll, \nll, \nll \rangle$ /* recur in child */ \;
                    \textbf{break}\;
                }
            }
        }
        
        \lIf{\isroot(P)}{\Return{}}
        \lElse{$\langle P,S,w,R \rangle \gets \langle \parent(P), P, \parind(P), \erre(P) \rangle$ /* backtrack */}
    }
}

\myproc{\isroot(X)}{
    \Return{$\parind(X)= \source(X)$}
}

\myproc{$\nextnode(w)$}{
\Return{$\min\{v\in \univ : v > w\}$}
}

\myproc{$\nextrestr(P,w,R)$}{
    \Return{the solution succeeding $R$ in the restricted problem $P\cup\{w\}$ (or \nll if $R$ is the last)}
}

\myproc{$\nextchild(P,w,R,S)$}{
        \ForEach{$x\in R : x > \source(S)$}{
            $C \gets \comp(\{x\},R)|_w$\;
            $D \gets \comp(C\cup \{w\}, \univ)$\;
            
            \If{$\langle \parent(D), \parind(D), \erre(D), \source(D)\rangle = \langle P,w,R,x\rangle$}{
                \Return{$D$}
            }
        }
        \Return{\nll}
}

\end{algorithm}

\section{Notes on set systems and properties}\label{apx:systems}

In this section we give some additional details about the set systems and properties mentioned in the paper, namely independence, strongly accessible and commutable set systems, and hereditary and connected hereditary properties.

Strongly accessible set systems contain all the other mentioned classes.
Strongly accessible set systems, in turn, a subset of the more general \emph{accessible set systems}, sometimes referred to as \emph{weakly accessible set systems}~\cite{boley2010listing}.
Recalling that set systems are defined as a family $\good$ of subsets of a $\univ$ set, weakly accessible set systems are those for which $X\in\good$ implies that if $X\neq \emptyset$ then $\exists x\in X$ s.t. $X\setminus\{x\}\in\good$; as in the other classes it is also assumed that $\emptyset\in\good$.

Hereditary properties and independence set systems  correspond to the same class, i.e., they both require that $X\in\good$ and $Y\subseteq X$ implies $Y\in\good$. It should be noted that Cohen et al.~\cite{cohen2008generating} defined such properties specifically on graphs, but remarked that different kinds of properties not on graphs could still be modelled as hereditary by representing the domain $\univ$ as the nodes of a graph with no edges.

Showing that set systems based on hereditary and connected hereditary properties form commutable set systems is straighforward.
We recall the \emph{commutable property}, that is given $X,Y\in \good$ with $Y\subset X$ and $a,b \in X\setminus Y$, if $Y\cup\{a\} \in \good$ and $Y\cup\{b\}\in\good$, then $Y\cup\{a,b\}\in \good$.
For hereditary properties, clearly $Y\cup\{a,b\}\in \good$ as it is a subset of $X$.
For connected hereditary the same holds: as $Y\cup\{a\}$ and $Y\cup\{b\}$ are two intersecting connected subgraphs of $X$, then $Y\cup\{a,b\}$ is also a connected subgraph of $X$, thus $Y\cup\{a,b\}\in \good$.

\end{document}